\date{}
\documentclass[11pt]{article}
\usepackage[all]{xy}

\usepackage{amsmath,amsthm,amsfonts,amssymb,amsopn,amscd}

\def\o{{\omega}}

\def\be{\begin{equation}}
\def\ee{\end{equation}}

\def\a{\alpha}

\def\e{\varepsilon}

\def\l{\lambda}

\def\setminus{\smallsetminus}
\def\gA{{\mathfrak A}}
\def\gB{{\mathfrak B}}

\def\A{{\cal A}}
\def\B{{\cal B}}
\def\C{{\cal C}}

\def\G{{\cal G}}

\def\N{{\cal N}}

\def\I{{\cal I}}
\def\H{{\cal H}}
\def\K{{\cal K}}
\def\S{{\cal S}}
\def\V{{\cal V}}
\def\f{{\varphi}}
\def\s{{\sigma}}

\def\S2{S^{1(2)}}
\newtheorem{theorem}{Theorem}[section]
\newtheorem{lemma}[theorem]{Lemma}

\newtheorem{corollary}[theorem]{Corollary}

\newtheorem{proposition}[theorem]{Proposition}

\theoremstyle{definition} 

\theoremstyle{remark} \newtheorem{remark}[theorem]{Remark}

\def\setminus{\smallsetminus}

\def\RR{{\mathbb R}}

\def\sl2{{{\rm SL}(2,\RR)}}
\def\psl2{{{\rm PSL}(2,\RR)}}

\def\u1{{{\rm V}(1)}}
\def\su2{{{\rm SV}(2)}}

\def\so3{{{\rm SO}(3)}}

\def\A{{\mathcal A}}
\def\B{{\mathcal B}}
\def\C{{\mathcal C}}

\def\F{{\mathcal F}}
\def\H{{\mathcal H}}
\def\I{{\mathcal I}}
\def\K{{\mathcal K}}

\def\N{{\mathcal N}}
\def\O{{\mathcal O}}

\def\V{{\mathcal V}}
\def\V{{\mathcal V}}

\def\G{{\bf G}}

\title{\Huge How to add a boundary condition
\\[0,8cm]
}

\author{
{\sc Sebastiano Carpi}\footnote{Supported in part by the ERC Advanced Grant 227458  
OACFT ``Operator Algebras and Conformal Field Theory" and GNAMPA-INDAM.}
\\
Dipartimento di Economia, Universit\`a di Chieti-Pescara ``G. d'Annunzio''\\
Viale Pindaro, 42, I-65127 Pescara, Italy \\
E-mail: {\tt s.carpi@unich.it}
\\[0,4cm]
{\sc Yasuyuki Kawahigashi}\footnote{Supported in part by Global COE Program ``The research and training center for new development in mathematics'', the Mitsubishi FoundationResearch Grants and the Grants-in-Aid for Scientic Research, JSPS.}\\
Department of Mathematical Sciences\\
University of Tokyo, Komaba, Tokyo, 153-8914, Japan\\[0,05cm]
           and\\[0,05cm]
Kavli Institute for the Physics and Mathematics of the Universe\\
5-1-5 Kashiwanoha, Kashiwa, 277-8583, Japan\\
E-mail: {\tt yasuyuki@ms.u-tokyo.ac.jp}
\\[0,4cm]
{\sc Roberto Longo}$^*$
\\
Dipartimento di Matematica,
Universit\`a di Roma ``Tor Vergata''\\
Via della Ricerca Scientifica, 1, I-00133 Roma, Italy\\
E-mail: {\tt longo@mat.uniroma2.it}}

\begin{document}

\maketitle
\vskip1cm

\begin{abstract}
Given a conformal QFT local net of von Neumann algebras $\B_2$ on the two-dimensional Minkowski spacetime with irreducible subnet $\A\otimes\A$, where $\A$ is a completely rational net on the left/right light-ray, we show how to consistently add a boundary to $\B_2$: we provide a procedure to construct a Boundary CFT net $\B$ of von Neumann algebras on the half-plane $x>0$, associated with $\A$, and locally isomorphic to $\B_2$. All such locally isomorphic Boundary CFT nets arise in this way. There are only finitely many locally isomorphic Boundary CFT nets and we get them all together. In essence, we show how to directly redefine the $C^*$ representation of the restriction of $\B_2$ to the half-plane by means of subfactors and local conformal nets of von Neumann algebras on $S^1$. 
\end{abstract}

\newpage

\section{Introduction}

The Operator Algebraic description of Boundary Conformal Quantum Field Theory has been set up in \cite{LR1} and turned out to be successful in various directions: not only it provides a simple, general, model independent framework but it allows a classification for the discrete series \cite{KL1,KLPR}, a method for constructing new Boundary QFT models \cite{LW,B, BT, Tan} and is suitable to describe various contexts \cite{LR4}. 

A Boundary CFT net $\B$ of von Neumann algebras on the half Minkowski plane $M_+$ associated with the chiral net $\A$ (that we always assume to be completely rational) is defined in Section {\ref{BCFT}. In particular, $\B$ is locally isomorphic to the restriction to $M_+$ of a local CFT net of factors on the two-dimensional Minkowski plane with chiral subnet $\A\otimes\A$ (the two-dimensional net on $M$, that is the product of the two chiral one-dimensional nets on the light-rays that we assume to be coincide with $\A$) and the restriction to $\A\otimes\A$ of the vacuum representation of $\B$ is to be induced by a representation of $\A$, see Subsection \ref{a+}.

One of the main results in \cite{LR1} is that Haag dual BCFT nets are in one-to-one correspondence with non-local, but relatively local, irreducible, one-dimensional extensions of the  local net $\A$ on the real boundary line. In particular, given such an extension $\B_0\supset \A$, the corresponding BCFT net $\B_0^{\rm ind}$ on $M_+$ is given by 
\[
\B_0^{\rm ind}(\O) = \B_0(K)'\cap\B_0(L)\ ;
\]
here $\O$ is the double cone $\O\equiv I\times J$ of $M_+$ (so $I,J$ are intervals of the boundary real line with disjoint closures and $J > I$), $L$ is the smallest interval containing $I$ and $J$ and $K$ is the interval between $I$ and $J$.

Despite its simplicity and effectiveness, the above holographic correspondence has some limitations. Firstly, only Haag dual BCFT nets appear and the non-dual ones are not directly visible, although present in the structure as intermediate subnets. Secondly, the local isomorphism of the BCFT net with a 2D CFT net on the full Minkowski plane can be seen quite indirectly through the computation of combinatorial data. Because of this last point, a more natural construction of the corresponding 2D net on the plane was later provided by shifting directly the boundary at infinity \cite{LR2}.

The converse problem, how to consistently add a boundary condition in any 2D CFT (without affecting the algebraic structure away from the boundary) has however so far remained unsettled. 

The purpose of this paper is indeed to provide a general algebraic procedure to this end. Given any 2D CFT net $\B_2$ on the Minkowski plane with chiral net $\A\otimes\A$, where $\A$ is completely rational, we construct a BCFT local net $\B$ on the half-plane based on $\A$ and locally isomorphic to $\B_2$.

Our procedure will generate all the BCFT local isomorphism class at one's hand. There are finitely many local, irreducible BCFT nets on $M_+$, locally isomorphic to $\B_2$, and all of them appear as irreducible subrepresentations of a reducible net that we are going to construct.

Our construction only depends on the restriction of $\B_2$ to $M_+$ so, starting with a Boundary CFT net on $M_+$, it also give a procedure to produce all the BCFT nets that are locally isomorphic to $\B$.

The reader willing to compare our Operator Algebraic construction to other works in more traditional approaches to Boundary CFT in Physics,  e.g. to the work of Cardy  \cite{Cardy1984,Cardy1989}, is referred to the discussion in the paper \cite{LR1}, where he can find motivations and illustrations. We say a couple of words here in the outlook.

We also mention that there are results, related to the ones in this paper, that were previously obtained by Kong and Runkel  within the context of modular tensor categories, see \cite{KR}.

\section{Preliminaries}
We shall denote by $M$ the Minkowski plane,\footnote{We deal here with two-dimensional spacetimes,  the reader may recognize how certain facts hold in more generality.} that is $\RR^2$ with the Lorentz metric $t^2 - x^2$. A double cone $\O$ in $M$ is a rectangle in the chiral coordinates, namely a subset of the form $\O=\{\langle t,x\rangle: u\equiv t+x\in I, v\equiv t-x\in J\}$ with $I,J$ open, non-empty intervals of $\RR$ and we set $\O\equiv I\times J$. We denote by $\K$ the set of double cones of $M$.

If $E\subset M$ is an open region, $\K(E)$ denote the set of double cones contained in $E$ with positive distance from the boundary of $E$. The half-plane $x>0$ is denoted by $M_+$ and in this case we shall put $\K_+ \equiv \K(M_+)$.

The M\"obius group $\G$ acts locally\footnote{If a Lie group $G$ acts on the manifold $P$ then, in particular, for every compact set $K\subset P$ there is a neighborhood $\V_K$ of the identity of $G$ such that $gK\subset P$ for all $g\in\V_K$. For the notion of local action see \cite{BGL}. }
on $\RR$ and globally on the one point compactification $S^1$ of $\RR$; the universal cover $\bar\G$ of $\G$ acts globally on the universal cover of $S^1$. The product action on the chiral lines of $M$ gives a local action of $\bar\G\times\bar\G$ on $M$, and a global action of (a quotient of) $\bar\G\times\bar\G$ on the Einstein universe (the cylinder). We shall be mainly concerned with the local action of $\G$ on $M$ obtained by restricting the local action of $\G\times\G$ to the diagonal, this action restricts to local actions of $\G$ on $M_+$ and on its boundary, the time-axis of $M$.

If $E$ is an open subset of $M$, the above local action of $\G$ on $M$ clearly restricts to a local action of $\G$ on $E$.

A \emph{net of von Neumann algebras $\B$ on $E$} is an isotonous map
\[
\B:\O\in\K(E)\mapsto \B(\O)
\]
where the $\B(\O)$'s are von Neumann algebras on a fixed Hilbert space $\H = \H_\B$.

$\B$ is \emph{local} if $\B(\O_1)$ and $\B(\O_2)$  commute if $\O_1$ and $\O_2$ are spacelike separated. In this paper all nets will be local.

Clearly, if $\B$ is a net on $E$ and $E_1\subset E$ we may restrict $\B$ to a net on $E_1$.

Suppose $G$ is a Lie subgroup of $\bar\G$, containing the translation subgroup, thus acting locally on $E\subset M$. The net $\B$ is $G$-covariant if there exists a unitary representation $U$ of (a covering of) $G$ on $\H_\B$ such that
\[
U(g)\B(\O)U(g)^* = \B(g\O)
\]
for every $\O\in\K(E)$ 
with \emph{$g,\O$ path connected to the identity within $E$} namely there exists a path of elements $g_s\in G$ connecting the identity of $G$ with $g$ such that $g_s\O\in\K(E)$. 

Moreover we assume the generator of the translation one-parameter subgroup to be positive and the existence of a unique (up to a phase) $U$-invariant unit vector $\Omega\in\H_\B$.

A \emph{representation} of the net $\B$ on a $E$ is a map
\[
\O\in\K(E)\mapsto \pi_\O
\]
where $\pi_\O$ is a normal representation of $\B(\O)$ on a fixed Hilbert space $\H_\pi$ such that
\[
\pi_{\tilde\O} |_{\B(\O)} = \pi_\O,\quad \O\subset\tilde\O\ .
\]

The representation $\pi$ of $\B$ on the Hilbert space $\H_\pi$ is covariant (with positive energy) if there exists a unitary representation $U_\pi$ on $\H_\pi$ of a covering group $\bar G$ of $G$ such that
\[
\pi_{g\O}(U(g)XU(g)^*)=U_\pi(g) \pi_\O(X)U_\pi(g)^*,\quad X\in\B(\O),\ \O\in\K(E), 
\]
with $g\in\bar G$ and $g,\O$ path connected to the identity within $E$

We shall say that two nets $\B_1$, $\B_2$ on $E$, acting on the Hilbert spaces $\H_1$ and $\H_2$, are \emph{locally isomorphic} if for every $\O\in\K(E)$ there is an isomorphism $\Phi_\O:\B_1(\O)\to\B_2(\O)$ such that
\[
\Phi_{\tilde\O} |_{\B_1(\O)} = \Phi_\O 
\]
if $\O$ and $\tilde\O$ belong to $\K(E)$ with $\O\subset \tilde\O$; \emph{$G$-covariantly isomorphic} if furthermore                
\[
 U_2(g)\Phi_\O(X) U_2(g)^* =\Phi_{g\O}(U_1(g)XU_1(g)^*),\quad X\in \B_1(\O)\ ,
\]
with $U_1$ and $U_2$ the corresponding unitary representation of $G$ on $\H_1$ and $\H_2$; and (unitarily) \emph{equivalent} if they are covariantly isomorphic and the isomorphism $\Phi_\O$ is implemented by a unitary $\H_1\to\H_2$ independent of $\O$.

Let $\B$ be a local net on $E$. We denote by $C^*(\B) = C^*(\B,E)$ the (locally normal)  \emph{universal C$^*$-algebra of $\B$}, that can be defined similarly as in chiral CFT context \cite{F,GL1,CCHW}: $C^*(\B)$ is the unique C$^*$-algebra with embeddings $\iota_\O:\B(\O)\to C^*(\B)$, $\O\in\K(E)$, such that $\iota_{\tilde\O}|_{\B\O)} =\iota_\O$ if $\O\subset\tilde\O$, generated by $\{\iota_\O(\B(\O)):\O\in\K(E)\}$, with a one-to one correspondence between locally normal representations of $\pi$ of $C^*(\B)$ and representations $\{\pi_\O\}_{\O\in\K(E)}$  of $\B$ on a Hilbert space $\H$ given by the commutative diagram
\[
\xymatrix{
C^*(\B) \ar[dr]^\pi  &  \\
\B(\O)\ar[r]^{\pi_\O}	\ar@{->}[u]^{\iota_\O}& B(\H) } 
\]
Here $\pi$ locally normal means that $\pi\cdot\iota_\O$ is normal, $\O\in\K(E)$. Moreover $C^*(\B)$ is minimal with this universality property: if $C^*(\B)^{\text{\tiny{$\sim$}}}$ is another $C^*$-algebra with the same properties, there exists a homomorphism $C^*(\B)^{\text{\tiny{$\sim$}}}\to C^*(\B)$ interchanging the embeddings of the $\B(\O)$'s.
Note that locally isomorphic nets have canonically isomorphic universal $C^*$-algebras.

We shall denote by $\gB(E)$ the $C^*$-algebra on the defining (vacuum) Hilbert space $\H_\B$ generated by the von Neumann algebras $\B(\O)$ as $\O$ runs in $\K(E)$. By the universality property, $\gB(E)$ is naturally a quotient of $C^*(\B,E)$, indeed
\be\label{quotient}
\gB(E) = \pi_0\big(C^*(\B,E)\big),
\ee
where $\pi_0$ is the lift to $C^*(\B,E)$ of the identity representation of $\B$.

If $F\subset E$ we have $\gB(F)\subset\gB(E)$. On the other hand, by the universally property, there a homomorphism
\[
C^*(\B,F)\to C^*(\B,E)\ ,
\]
cf. the general covariance principle \cite{BFV}.
\begin{lemma}\label{simple}
If $\K(E)$ is an inductive family and the $\B(\O)$ are type $III$ factors (on a separable Hilbert space), then $C^*(\B,E)$ is equal to $\gB(E)$ and is a simple, purely infinite $C^*$-algebra in Cuntz standard form (cf. \cite{I,CCHW}).
\end{lemma}
\begin{proof}
$\gB(E)$ is the $C^*$ inductive limit of the type $III$ factors $\B(\O)$. As the $\B(\O)$'s are simple $C^*$-algebras, also $\gB(E)$ is simple and it has the above universality property so $C^*(\B)= \gB(E)$. As a type $III$ factor is purely infinite in Cuntz standard form, so is $\gB(E)$.
\end{proof}
The following notion will be crucial in this paper.
\smallskip

\noindent
Let $\gA\subset\gB$ unital C$^*$-algebras. If $\gB$ acts on a Hilbert space $\H$, we shall say that $\xi\in\H$ is a \emph{bicyclic vector} for $\gA\subset\gB$ if $\xi$ is cyclic for $\gA$, hence for $\gB$. If $\f$ is a state on $\gB$, we shall say that $\f$ is a 
\emph{bicyclic state} for $\gA\subset\gB$ if the vector $\xi_\f$ in the GNS representation $\pi_\f$ of $\gB$ is bicyclic for $\pi_\f(\gA)\subset\pi_\f(\gB)$.
\begin{lemma}\label{bicyclic}
Let
\[
\begin{array}{ccc}N & \subset & M \\ \cup &  & \cup \\N_0 & \subset & M_0\end{array}
\]
be a square of inclusions of infinite factors and $\varepsilon:M\to 
M_0$
a normal faithful conditional expectation with $\varepsilon(N)=N_0$. 
Suppose that
$N_0\subset N$ and $M_0\subset M$ are irreducible inclusions with the 
same
finite index: $[M:M_0]=[N:N_0]$.

If $\f_0$ is normal state on $M_0$ which is bicyclic for $N_0\subset M_0$, then the state $\f\equiv \f_0\cdot\e$ on $M$ is bicyclic for $N\subset M$.
\end{lemma}
\begin{proof}
We can assume we are in the Hilbert space $\H$ of the GNS representation of $\f$ with cyclic vector $\xi$. By assumption $\H_0\equiv\overline{N_0\xi}=\overline{M_0\xi}$.

Now, following the argument in \cite[Prop. 2.3]{GL1}, there is an isometry $T\in N$ such that $N =N_0 T$ and $M=M_0 T$, cf. \cite{Lhopf,Linterm}. Therefore $\H= \overline{M\xi} = \overline{T^*M_0\xi}=  \overline{T^*\H_0}=  \overline{T^*N_0\xi}$, namely $\f$ is bicyclic for $N\subset M$.
\end{proof}
\\[-1cm]
\section{Nets on $M_+$ and on a wedge}
\label{Boundary QFT}

We now discuss more explicitly the nets of von Neumann algebras on the half-plane $M_+$, on a wedge region $W$, and their relations.
\subsection{Nets on $M_+$}
Let $I_1 , I_2$ be intervals of time-axis and $\O= I_1 \times I_2$ be the double cone of $M$ associated with $I_1 , I_2$. Then $\O\subset M_+$ if{f} $I_2 > I_1$ and in this case $\O\in\K_+$  if{f} the closures of $I_1$ and $I_2$ have empty intersection. 

Let $G$ be a Lie subgroup of $\G$ with the inherited local action on $M$, hence on $M_+$, by the diagonal action as above. We will always assume that $G$ contains the translation one-parameter subgroup.
\medskip

\noindent
\emph{A (local, $G$-covariant) QFT net $\B$ of von Neumann algebras on $M_+$} on a Hilbert space $\H = \H_\B$ is a triple $(\B, U, \Omega)$ where
\begin{itemize}
\item $\B$ is a isotonous map 
\[
\O\in\K_+\mapsto \B(\O)\subset B(\H)
\] 
where $\B(\O)$ is a von Neumann algebra on $\H$;
\item
$U$ is a unitary representation of $G$ on $\H$, with positive generator for translation one-parameter subgroup, such that
\[
U(g)\B(\O)U(g)^* = \B(g\O),\quad \O\in\K_+ ,
\]
for all $g\in G$ such that $g,\O$ is path connected to the identity within $M_+$  (in particular for all $g$ in the translation subgroup).
\item
$\Omega\in\H$ is a unit vector such that $\mathbb C\Omega$ are the $U$-invariant vectors and $\Omega$ is cyclic for $\B(\O)$ for each fixed $\O\in\K_+$.
\item $\B(\O_1)$ and $\B(\O_2)$ commute if $\O_1,\O_2\in\K_+$ are spacelike separated.
\end{itemize}
We shall shortly indicate the triple $(\B, U,\Omega)$ by $\B$.
\begin{proposition} 
\label{sep}
Let $\B$ be a local, $G$-covariant net as above. Then $\Omega$ is separating for every $\B(\O)$, $\O\in\K_+$, and $\B$ is irreducible, namely $\gB(M_+)$ is irreducible.
\end{proposition}
\begin{proof}
If $\O\in\K_+$ there is $\O_1\in\K_+$ space like to $\O$. By locality $\B(\O_1)\subset\B(\O)'$, so $\Omega$ is cyclic for $\B(\O)'$, thus separating for $\B(\O)$.

Concerning the irreducibility, since the one-parameter translation group of unitaries $U(t)$ has positive generator, $U(t)\Omega=\Omega$ and  $U(t)\gB(M_+)U(t)^*=\gB(M_+)$, we have $U(t)\in\gB(M_+)''$ by a theorem by Borchers as $\Omega$ is cyclic, see \cite{LN}. So the one-dimensional projection onto  $\mathbb C\Omega$  belongs to $\gB(M_+)''$ and this implies $\gB(M_+)''=B(\H)$, see \cite{LN}.
\end{proof}
A representation of the QFT net $\B$ of von Neumann algebras on $M_+$ on a Hilbert space $\H$ is a map
\begin{equation}\label{pi}
\pi: \O\in\K_+\mapsto \pi_\O
\end{equation}
where $\pi_\O$ is a normal representation of $\B(\O)$ on $\H$ and $\pi_{\tilde O}|_\O =\pi_\O$ if $\O\subset\tilde\O$.

$\pi$ is covariant (with positive energy) if there exists a unitary representation $U_\pi$ of $G$ on $\H$ (with translation positive generator) s.t.
\begin{equation}\label{tcov}
U_\pi(g)\pi_\O(X)U_\pi(g)^* = \pi_{g\O}\big(U(g) X U(g)^*\big)
\end{equation}
if $X\in \B(\O)$ with $g\in G$, $\O\in \K_+$ and $g,\O$ is path connected to the identity within $M_+$.
\subsection{Nets on the wedge $W$}
We shall always denote by $W$ the right wedge $\{\langle t,x\rangle : x>|t|\}$. Thus $W\subset M_+$.

If $\B$ is a $G$-covariant QFT net of von Neumann algebras on $M_+$, clearly $\B$ restricts to a $G$-covariant QFT net of von Neumann algebras on $W$, which is $G$-covariant w.r.t. the local action of $G$ on $W$. Conversely:
\begin{lemma}\label{wd}
Let $\B_W$ a local, $G$-covariant QFT net of von Neumann algebras on $W$ on a Hilbert space $\H$. There exists a unique local, $G$-covariant QFT net $\B$ of von Neumann algebras on $M_+$ on $\H$ whose restriction to $W$ is $\B_W$.
\end{lemma}
\begin{proof}
With $U$ the unitary representation of (a covering of) $G$ associated with $\B_W$, we set
\be\label{ext}
\B(g\O)\equiv U(g)\B_W(\O)U(g)^*, \quad \O\in\K(W)\ ,
\ee
with $g,\O$ path connected to the identity within $M_+$. Note that all double cones in $\K_+$ have this form, indeed every $\O\in \K_+$ is a translated of a double cone in $\K(W)$.

It is enough to check that \eqref{ext} well defines $\B$, as then  the net properties of $\B$ are then immediate. So let $\O_1\in\K_+$ and $g_1\in G$, with $g_1,\O_1$ be path connected to the identity within $M_+$ and satisfy $U(g_1)\B_W(\O_1)U(g_1)^*= U(g)\B_W(\O)U(g)^*$. Then
\[
U(h)\B_W(\O)U(h)^*  = \B_W(\O_1), \quad h\equiv g^{-1}_1 g\ .
\]
We show now that $h$ is unique. Indeed, as an element of $\bar\G$, the equation $h\O=\O_1$ determines $h$ uniquely up to a M\"obius $2\pi$-rotation (on the boundary line an element of $\G$ is determined by its action on four different points). But then $h$ is unique by the path connectedness property. So $g\O=g_1\O_1$. 

The locality of $\B$ follows because, if $\O_1 , \O_2$ a pair of spacelike separated double cones in $\K_+$, there is a pair  ${\O_1}_W ,{ \O_2}_W$ of spacelike separated double cones in $\K(W)$ and a time-translation mapping ${\O_1}_W ,{ \O_2}_W$ onto $\O_1 , \O_2$.
\end{proof}
\begin{proposition}\label{extpi}
Let $\B$ a local, $G$-covariant QFT net of von Neumann algebras on $M_+$ and $\B_W$ its restriction to $W$.

Every (covariant) representation of $\B$ restricts to a (covariant) representation of $\B_W$.

Conversely, every covariant representation of $\B_W$ extends uniquely to a covariant representation of $\B$.
\end{proposition}
\begin{proof}
The first statement is obvious.

As for the second statement, let $U$ be the covariant unitary representation of $G$.
Let $\pi$ be a covariant representation of $\B_W$. We set
\be
\tilde\pi_{g\O}\big(U(g)XU(g)^*\big) \equiv U_\pi(g)\pi_\O(X)U_\pi(g)^*,\quad X\in\B_W(\O),\ \O\in\K_W\ ,
\ee
where $g,\O$ is path connected to the identity within $M_+$ and
$U_\pi$ is the covariant unitary representation of $G$ in the representation $\pi$.

By the argument in Lemma \ref{wd}, $\tilde\pi_{g\O}$ is well defined. It is immediate to check that $\tilde\pi: \O\in\K_+\mapsto \tilde\pi_{\O}$ is a covariant representation of $\B$ extending $\pi$.
\end{proof}
Let $\B$ a $G$-covariant net on $W$ as above and $\gB(W)$ the the associated $C^*$-algebra (isomorphic to $C^*(\B_W)$). If $g\in G$, $\O\in\K(W)$ with $g,\O$ path connected to the identity within $W$, we denote by $\a_g^\O$ the isomorphism of $\B(\O)$ onto $\B(g\O)$ given by
\be\label{invariant}
\a_g^\O (X) = U(g)XU(g)^*,\quad X\in\B(\O)\ .
\ee
A state $\f$ of $\gB(W)$ is \emph{$\a$-invariant} if 
\[
\f(\a_g^\O(X)) =\f(X),\quad X\in\B(\O),
\]
for all $g\in G$ and $\O\in\K(W)$ with $g,\O$ path connected to the identity within $W$.
\begin{lemma}\label{Wcov}
Let $\f$ be an $\a$-invariant state of $\gB$ as above and suppose $\f$ to be bicyclic for $\B(\O)\subset\B(\tilde\O)$ if $\O\subset\tilde\O\in\K(W)$. Then the GNS representation $\pi_\f$ of $\gB$ is covariant, indeed there exists a unique unitary representation $U_\f$ of (a cover of) $G$ implementing the covariance of $\pi_\f$ as in eq. \eqref{tcov} and $U_\f(g)\xi_\f =\xi_\f$.
\end{lemma}
\begin{proof}
Note that $\xi_\f$ is cyclic for $\pi_\f(\B(\O)$ for every $\O\in\K(W)$. Indeed we can find an increasing sequence of double cones $\O_n\in\K(W)$ containing $\O$ such that $\cup_n \O_n = W$. Then
\[
\overline{\pi_\f(\B(\O))\xi_\f} = \overline{\pi_\f(\B(\O_n))\xi_\f} = \overline{\bigcup_k\pi_\f(\B(\O_k))\xi_\f}=
\overline{\pi_\f(\gB(W))\xi_\f} = \H_\f \ .
\]
Let $\O\in\K(W)$ and $\V_\O$ a connected neighborhood of the identity of $G$ such that $g\O\in\K(W)$ for all $g\in\V_\O$. For each fixed $g\in\V_\O$, the formula 
\[
U_\f(g)\pi_\f(X)\xi_\f = \pi_\f(\a_g^\O(X)), \quad X\in\B(\O),
\]
defines an isometric map $U_\f(g)$ from $\pi_\f(\B(\O))\xi_\f$  to  $\pi_\f(\B(g\O))\xi_\f$, thus a unitary on $\H_\f$. 

One can immediately see that $U(gh) = U(g)U(h)$ if $g,h, gh$ belong to $\V_\O$, so $U$ extends to a unitary representation of a covering of $G$, independent of $\O$, implementing the covariance. The uniqueness is clear.
\end{proof}
\\[-1cm]
\subsection{The net $\A_+$ and its induced representations}
\label{a+}

We first recall the definition of a local M\"obius covariant  of von Neumann algebras on $\mathbb R$.
See \cite{LN} for more details. We denote by $\G$ the M\"obis group, that naturally acts on $S^1$.

\subsubsection{M\"obius covariant nets}

A \emph{local, M\"obius covariant net  of  von Neumann algebras $\A$ on $S^1$} is a map
$
\A:I\to\A(I)
$
from $\I$, the set of (open, proper) intervals of $S^1$, to the set of von Neumann algebras on a Hilbert space $\H$ that verifies the following properties:
\begin{description}
\item$\textnormal{\textit{Isotony:}}$ {If $I_1$, $I_2$ belong to $\I$  and $I_1\subset I_2$, then
$
\A(I_1)\subset\A(I_2)\ .
$}
\end{description}
\begin{description}
\item\textnormal{\textit{M\"obius covariance:}}
{There is a
strongly continuous unitary representation $U$ of ${\bf G}$ on $\H$ such that
$
U(g)\A(I)U(g)^*=\A(gI), \ g\in {\bf G},\ I\in\I.
$}
\end{description}
\begin{description}
\item
\textnormal{\textit{Positivity of the energy:}}
{$U$ is a positive energy representation.} 
\end{description}
\begin{description}
\item\textnormal{\textit{Existence and uniqueness of the vacuum:}}
{There exists a unique (up to a phase) unit $U$-invariant vector $\Omega$ 
(vacuum) and $\Omega$ is
cyclic for the von Neumann algebra $\vee_{I\in\I}\A(I)$}
\end{description}
\begin{description}
\item$\textnormal{\textit{Locality:}}$
{If $I_1$ and $I_2$ are disjoint intervals, the von Neumann algebras $\A(I_1)$ and $\A(I_2)$ commute:}
$
\A(I_1)\subset\A(I_2)'
$
\end{description}
The uniqueness (up to a phase) of the vacuum vector $\Omega$ is equivalent to the irreducibility of $\A$ (i.e. $\big(\cup_I \A(I)\big)''=B(\H)$) and to the factoriality of the local von Neumann algebras $\A(I)$, see \cite{BGL, LN}.

\medskip

\noindent
A \emph{local M\"obius covariant net on $\RR$} is the restriction of a local M\"obius covariant net on $S^1$ to $\RR = S^1\setminus \{-1\}$ (identification by the stereographic map).
\medskip

\noindent
Concerning completely rational nets, we refer the reader to \cite{KLM}.
\medskip

\subsection{Inducing representations}

Let then $\A$ be a local M\"obius covariant net of von Neumann algebras on $\RR$, that we shall always assume to be \emph{completely rational}, on a Hilbert space $\H=\H_\A$. If $L\subset\RR$ is an open, non-empty subset, we shall denote by $\gA(L)$ the $C^*$-algebra generated by $\A(I)$ as $I$ runs in the intervals contained in $L$.

By a \emph{locally normal representation} $\pi$ of $\gA(\RR)$ on a Hilbert space $\H_\pi$ we shall mean a representation $\pi$ of $\gA(\RR)$ on $\H_\pi$ such that $\pi |_{\A(I)}$ is normal for all intervals $I$ of $\RR$. (If $\H_\pi$ is separable, every representation of $\gA(\RR)$ is locally normal.)

By a \emph{representation} (or DHR representation) $\pi$ of $\gA(\RR)$ we shall mean as usual a consistent family of normal representations $\pi_I$ of $\A(I)$ on $\H_\pi$ associated to intervals of $S^1 = \RR\cup\{\infty\}$. 

There is a natural correspondence between a representation $\pi$ of $\A$ and a representation of $\pi$ of $\gA(\RR)$ such that $\pi|_{\gA(\RR\setminus I)}$ is normal for every open non-empty interval $I$ of $\RR$, see \cite[Appendix]{KLM}.
 
Starting with the local, M\"obius covariant net $\A$ on $\H$, we have a local $\G$ covariant net $\A_+$ on $M_+$ defined by
\be\label{A_+}
\A_+(\O) \equiv \A(I)\vee\A(J), \quad \O=I\times J \in \K_+\ .
\ee
Note that by the split property $\A_+(\O)$ is naturally isomorphic to $\A(I)\otimes\A(J)$, and we have a left identification of $\A(I)$ with the subalgebra $\A(I)\otimes \mathbb C$ of $\A(\O)$ and a right identification of $\A(J)$ with the subalgebra $\mathbb C\otimes \A(J)$ of $\A(\O)$.

The unitary representation $U$ of the M\"obius group $\G$ giving the covariance for $\A$ also gives the $\G$-covariance for $\A_+$:
\[
U(g)\A_+(\O)U(g)^* = \A_+(g\O)
\]
where $\O=I\times J$ and $g\O\equiv gI\times gJ$.

Note that the vacuum vector $\Omega$ is bicyclic for $\A(I)\subset\A_+(\O)$ and for $\A(J)\subset\A_+(\O)$ with the left and right identification.

We explicitly recall  the following basic fact.
\begin{proposition}
The net $\A_+$ is covariantly locally isomorphic to $(\A\otimes \A)|_{M_+}$, the restriction to $M_+$ of the 2D net  $\A\otimes \A$ on $M$ with covariance unitary representation $U\otimes U$ (diagonal action of $\G\subset\G\times\G$).
 $\A_+$ and $(\A\otimes \A)|_{M_+}$ are not unitarily equivalent.
 \end{proposition}
 \begin{proof}
The local isomorphism follows by the split property of $\A$, see \cite{DL}. That $\A_+$ and $(\A\otimes \A)|_{M_+}$ are not equivalent follows, for example, from the fact that the split property fails for contiguous intervals so the von Neumann algebra associated with $\O=(-1,0)\times (0,1)$ by the net $\A_+$ is not naturally isomorphic with $\A(-1,0)\otimes\A(0,1)$.
 \end{proof}
 \\[-0,7cm]
Therefore $\A_+$ and $(\A\otimes \A)|_{M_+}$ share the same representations, indeed one could consider one a representation of the other.
\medskip

\noindent
If $\pi$ is a locally normal representation of $\gA(\RR)$ on a Hilbert space $\H_\pi$, there is an induced representation $\pi_+$ of $\A_+$ on $\H_\pi$ given by
\[
\pi_+ {_\O}(XY) \equiv \pi(XY), \quad X\in\A(I),\ Y\in \A(J),\quad \O=I\times J \in \K_+\ .
\]
Clearly $\pi_+$ is $\G$-covariant with positive energy if so is $\pi$.
Namely there is a unitary representation $U_\pi$ of $\G$ with positive energy such that
\[
U_{\pi}(g){\pi}_{\O}(X)U_{\pi}(g)^* = {\pi_+}_{g\O}(U(g)XU(g)^*),\quad X\in\A_+(\O)\ ,
\]
with $U$ the unitary covariance representation of $\G$.

Note that this is always the case if $\pi$ is a (DHR) representation of $\A$,  
because every representation of a completely rational net is $\G$-covariant with positive energy.

If $\pi_0$ is the vacuum (i.e. identity) representation of $\A$, we call ${\pi_0}_+$ the \emph{vacuum representation of} $\A_+$.
\begin{lemma}\label{lemma1}
Let $\pi$ be a representation of $\A_+$ on a Hilbert space $\H$. The following are equivalent:

$(i)$: $\pi$ is $\G$-covariant with positive energy and there is a  $U_{\pi}$-invariant unit vector $\Omega\in\H$ which is bicyclic for $\pi_\O(\A(I))\subset\pi_\O(\A_+(\O))$ and for $\pi_\O(\A(J))\subset\pi_\O(\A_+(\O))$, $\O\equiv I\times J$, with the left and right identification.

$(ii)$: $\pi$ is unitarily equivalent to the vacuum representation ${\pi_0}_+$.
\end{lemma} 
\begin{proof}
$(ii)\Rightarrow (i)$ is immediate by the Reeh-Schlieder property of the vacuum vector for $\A$, see \cite{LN} or \cite{BGL}.

$(i)\Rightarrow (ii)$:
With $I$ an interval of $\RR$, set $\C(I)\equiv\A_\pi(W_I)$ where $W_I \equiv \{\langle t,x\rangle :\, t\pm x\in I, \, x>0\}$ is the left wedge of $M_+$ spanned by $I$ and $\A_\pi(W_I)$ is the von Neumann algebra generated by $\pi_\O(\A_+(\O))$ as $\O\in\K_+$ runs within $W_I$. Note that $\C$ is a net, $\Omega$ is cyclic for each $\C(I)$  because $ \A_\pi(W_I)\supset \pi_\O(\A_+(\O))$ if $\O\subset W_I$, separating for each $\C(I)$ because
$ \A_\pi(W_I)$ commutes with $\pi_\O(\A_+(\O))$ if $\O\subset W'_I$, and $U_\pi$ implements the $\G$-covariance for $\C$. Therefore $\C$ is a M\"obius covariant net on $\RR$ (we shall soon show that $\C$ is local and irreducible), in its vacuum representation.

Now we set 
\[
\pi(\A(I))\equiv \bigcap_{J}\pi_\O(\A(\O)), \quad \O=I\times J,
\]
(intersection over all intervals of $\RR$). By the split property $\pi(\A(I))$ is naturally isomorphic to $\A(I)$.

Nowt $\pi(\A(I))\subset\C(I)$ for every open interval $I$; in fact if $I_0\Subset I$ is an open interval with $I_0\subset I$, then $\pi(\A(I_0))\subset\C(I)$ with the left identification because $I_0\times J_\e\subset W_I$ for a suitable small interval $J_\e$. 

Now the map $I\mapsto \pi(\A(I))$ is a M\"obius covariant net on $\H$ in the vacuum representation, as $\Omega$ is cyclic for it by the bicyclicity assumption. Moreover $\pi(\A(I))$ is local because if $I_1, I_2$ are disjoint intervals there is an interval $I\supset I_1, I_2$ and $\pi(\A(I_1))$ and $\pi(\A(I_2))$ are commuting subalgebras of $\pi(\A(I))$.

So, by the bicyclicity property of $\Omega$, we have $\pi(\A(I))=\C(I)$ with the left (and analogously with the right) identification; this follows by the Bisognano-Wichmann property for $\A$, see \cite{BGL}, and the Tomita-Takesaki theory ($\C(I)'$ is a cyclic von Neumann subalgebra of $\pi(\A(I))'$ globaly invariant under the modular group action).

Then the net $\C_+$ is given by
\[
\C_+(\O) = \C(I)\vee\C(J) = \pi(\A(I))\vee\pi(\A(J)) = \pi_\O(\A(\O)), \quad \O=I\times J,
\]
namely $\pi ={\pi_0}_+$ where $\pi_0$ is the vacuum representation of $\C$; as $\C$ is naturally isomorphic to $\A$, the proof is completed.
\end{proof}
If $E\subset M_+$, we denote by $\gA_+(E)$ the $C^*$-algebra on $\H_\A$ associated with $E$ by $\A_+$ as in \eqref{quotient}.
\begin{lemma}\label{irrW}
We have:

$(a)$: If $W_1 =(-\infty, -\delta_1)\times (\delta_2\times \infty)$ is a wedge contained in $W$, $\delta_1, \delta_2 \geq 0$, then $\gA_+(W_1)$ is the $C^*$-algebra generated by $\gA(-\infty, -\delta_1)$ and $\gA(\delta_2\times \infty)$.

$(b)$: $\gA_+(W_1)$ is a simple $C^*$-algebra.

$(c)$: If $\pi$ is a locally normal representation of $\gA(\RR)$ then $\pi_+(\gA_+(W))''=\pi(\gA(\RR))''$.
\end{lemma} 
\begin{proof}
$(a)$ is immediate by the definition of $\A_+$. 

$(b)$ holds true as  $\gA_+(W_1)$ is the inductive limit of the $\A_+(\O)$'s, $\O\in\K(W_1)$ and the each $\A(\O)$ is simple because it is a type III factor.

$(c)$ We have $\gA_+(W) = \gA(\RR\setminus\{0\})$ so $\pi_+(\gA_+(W))'' = \pi(\gA(\RR\setminus\{0\}))''\subset\gA(\RR)''$. So we have to show that $\pi(\gA(\RR\setminus\{0\}))''=\gA(\RR)''$, namely that
$\pi(\gA(\RR\setminus\{0\}))''\supset\pi(\A(I))$ if $I$ is an interval of $\RR$ and $0\in I$. But this is true because by strong additivity $\gA(I\setminus \{0\})'' =\A(I)$, so $\pi(\gA(I\setminus \{0\}))'' =\pi(\A(I))$ by the local normality of $\pi$.
\end{proof}
\begin{proposition}\label{repA_+}
Let $\s$ be a representation of $\A_+$ on a Hilbert space $\H$. Suppose there is a unit vector $\Omega\in\H$ such that the state $\o\equiv (\Omega,\cdot\ \Omega)$ is bicyclic for $\s_\O(\A(I))\subset\s_\O(\A_+(\O))$ and for $\s_\O(\A(J))\subset\s_\O(\A_+(\O))$, for every double cone $\O=I\times J$, with the left and right identification, and that $\o$ is bicyclic for $\s_{\tilde\O}(\A(\O))\subset\s_{\tilde\O}(\A(\tilde\O))$ too, if $\O\subset\tilde\O$ belong to $\K(W)$.

Let $\H_0$ be the Hilbert subspace $\overline{\s_\O(\A_+(\O))\Omega}$ of $\H$ (independent of $\O$), $\s_0$ the restriction of $\s$ to $\H_0$ and assume that $\s_0$ is $\G$-covariant with positive energy and $\Omega$ a is a fixed $U_{\s_0}$-vector.

If $\Omega$ is separating for $\s(\gA_+(W_1))''$, for any wedge $W_1\subset W$ with positive distance to the boundary of $W$, then $\s$ is equivalent to $\pi_+$ with $\pi$ a (DHR) representation of $\A$.
\end{proposition} 
\begin{proof}
By Lemma \ref{lemma1}, the restriction $\s_0$ of $\s$ to $\H_0$ is equivalent to ${\pi_0}_+$ where $\pi_0$ is the vacuum representation of $\A$.

The restriction map $r_\O:\s_\O(X) \mapsto {\s_0}_\O(X)\equiv \s_\O(X)|_{\H_0}$, $X\in\A_+(\O)$, is normal and one-to-one because $\Omega$ is separating for $\s(\A_+(\O))$ and we have
\[
\s_\O = r^{-1}_\O\cdot {\s_0}_\O = r^{-1}_\O\cdot {{\pi_0}_+}_\O = {\pi_+}_\O \ ,
\]
where $\pi$ is the locally normal representation of $\gA(\RR)$ determined by $\pi(X)|_{\H_0} = \pi_0(X)$, $X\in\gA(\RR)$.

We want to show that $\pi$ is a DHR representation of $\A$.

The restriction map $\s(X) \mapsto \s_0(X)\equiv \s(X)|_{\H_0}$, $X\in\gA_+(W_1)$, extends to a normal map $r_{\text{\tiny $W_1$}}:\s(\gA_+(W_1))'' \to \s_0(\gA_+(W_1))''$ which is invertible because $\Omega$ is separating for $\s(\gA_+(W_1))''$. We have:
\be\label{ln}
\s|_{\gA_+(W_1)} = r_{\text{\tiny $W_1$}}^{-1}\cdot\s_0 |_{\gA_+(W_1)} \simeq  r_{\text{\tiny $W_1$}}^{-1}\cdot{\pi_0}_+ |_{\gA_+(W_1)} =  r_{\text{\tiny $W_1$}}^{-1}\cdot{\pi_0}|_{\gA(I_1\cup J_1)}
\ee
where $I_1$ and $J_1$ are the left and right half-lines such that $W_1 = I_1\times J_1$.

Therefore $\pi\equiv r^{-1}\cdot{\pi_0}$ is a locally normal representation of $\gA(\RR)$ such that such that
$\s = \pi_+$ and, by eq. \eqref{ln},  $\pi$ is normal on $\gA(I_1\cup J_1)$, hence $\pi$ is a DHR representation, see the appendix of \cite{KLM}.
\end{proof}
\section{Boundary CFT nets and their characterization}\label{BCFT}
A \emph{Boundary CFT net} of von Neumann algerbras on $M_+$ was defined in \cite{LR2}, here we give an equivalent, slightly different, definition (see Remark \ref{rem1}).
\smallskip

\noindent
Let $\A$ be a M\"obius covariant local net of factors on $\RR$.
\medskip

\noindent
A local {\em Boundary  CFT (BCFT)} net $\B$ associated 
with $\A$ is a local, isotonous map 
\[
\B:\O\in\K_+\mapsto \B(\O),
\]
where the $\B(\O)$ are von Neumann algebras on a Hilbert space $\H_\B$, such that
\smallskip

$(a)$ there is a unitary representation $U$ of the covering of the
M\"obius group $\bar\G$ with positive generator for
the subgroup of translations, such that  
\begin{equation} 
U(g)\B(\O)U(g)^* = \B(g\O) 
\end{equation}
whenever $g\in \bar\G$ takes the double cone
$\O=I\times J\in\K_+$ into another double-cone $g\O\equiv gI
\times gJ$ within $M_+$, with an
invariant vector unit $\Omega\in\H_\B$ (the vacuum vector). 
\smallskip

$(b)$  There is a representation $\pi$ of $\A$ on $\H_\B$ such that
$\B_+(\O)$ contains $\pi_+(\A_+(\O))$ with $\pi_+(\A_+(\O))'\cap\B(\O)=\mathbb C$
and 
\begin{equation} 
U(g)\pi_+(\A_+(\O))U(g)^* = \pi_+(\A_+(g\O))
\end{equation}
whenever $\O$ and $g\O$ belong to $\K_+$. 
\smallskip

$(c)$ The vector $\Omega$ is cyclic and separating for each $\B(\O)$, $\O\in\K_+$.
\medskip

\noindent
We shall say that $\B$ is irreducible if the $C^*$-algebra $\gB(M_+)$ is irreducible.
\begin{remark}\label{rem1}
 In \cite{LR2} a local, irreducible BCFT net $\B$ associated as above but the irreducibility of the inclusion  $\pi_+(\A_+(\O))'\cap\B(\O)=\mathbb C$ is not assumed there, moreover $\Omega$ is only assumed to be cyclic or separating although it has to be cyclic for $\gB(M_+)$ which is irreducible. On the other hand \emph{joint irreducibility} is assumed there:
\[
\pi_+(C^*(\A_+))''\vee\B(\O) = B(\H)
\]
for any double cone $\O\in\K_+$.

As shown in \cite{LR2}, the definition of a BCFT net implies the above properties $(a),(b),(c)$. 

On the other hand, if we assume the above properties $(a),(b),(c)$ and that $U(g)\in \pi_+(C^*(\A_+))''$, $g\in\bar\G$, then 
$\B$ is a BCFT net in the sense of \cite{LR2}.

While we shall deal here with the above definition, we shall obtain the property that $U(g)\in \pi_+(C^*(\A_+))''$ in Corollary \ref{posen}, so our constructed nets will be BCFT nets in the sense of \cite{LR2}.
\end{remark}
\noindent
We now give an equivalent formulation for a boundary CFT net.

\begin{proposition}
Let $\B$ a (local) BQFT net on $M_+$ that we assume to be $\bar\G$-covariant with positive energy. We assume that $\B$ contains $\A_+$ as an irreducible subnet, where $\A_+$ is the net on $M_+$ associated to a chiral local M\"obius covariant net $\A$ as above \eqref{A_+}. Suppose the vacuum vector $\Omega$ of $\B$ gives a bicyclic state for $\A(I)\subset\A_+(\O)$ and $\A(J)\subset\A_+(\O)$ with the left and right embedding, $\O= I\times J$.

Then $\B$ is a Boundary CFT net on $M_+$associated with $\A$.
\end{proposition}
\begin{proof}
By assumption there is a representation $\s$ of $\A_+$ such that $\s(\A_+(\O))$ is an irreducible subfactor of $\B_+(\O)$ and the covariance of $\s$ is obtained from the covariance of $\B_+$ and we need to show that there is a representation $\pi$ of $\A$ such that $\pi_+ =\s$.

We then obtain the conclusion by Lemma \ref{lemma1} once we know that the vacuum vector $\Omega$ is separating for $\s(\gA_+(W_1))''$ if $W_1$ is a right translated of the wedge $W$. Now $\Omega$ is cyclic for $\B(\O)$ if $\O\in\K_+$, $\O\subset W'_1$, hence for $\B_+(W'_1)$, therefore by locality $\Omega$ is separating for $\B(W_1)$. Since $\s(\gA_+(W_1))''$ is contained in $\B(W_1)$, $\Omega$ is separating for $\s(\gA_+(W_1))''$ as well.
\end{proof}
\\[-1,3cm]
\begin{proposition}
Let $\B$ a local, Boundary CFT net on $M_+$ associated with $\A$ as above with $\A$ completely rational.
For each $\O\in\K_+$ there exists a unique conditional expectation $\e_\O:\B(\O)\to\pi_+(\A_+(\O))$ and the family is consistent: $\e_{\tilde\O}|_{\B(\O)} =\e_\O$ if $\O\subset\tilde\O$ belong to $\K_+$.
\end{proposition}
\begin{proof}
Let $\O\in\K_+$. Then $\pi_+(\A_+(\O))\subset\B(\O)$ is irreducible with finite index \cite{LR1} so there exists a unique conditional expectation $\e_\O:\B(\O)\to\pi_+(\A_+(\O))$. To check the consistency, it is enough to assume that $\tilde\O\supset\O$ has one corner in common with $\O$ (by iterating the argument). So $\tilde\O = \tilde I\times \tilde J$ and $\O =  I\times J$ and $I_1\equiv\tilde I\setminus I$, $J_1\equiv\tilde J\setminus J$ are (not open) intervals, so $\O'\cap\tilde\O = \O_1$ where $\O_1\equiv I_1\times J_1 \in\K_+$. Then 
\begin{multline*}
\A_+(\O_1)'\cap\A_+(\tilde\O) = \big(\A(I_1)\vee\A(J_1)\big)'\cap \big(\A(\tilde I)\vee\A(\tilde J)\big) \\
= \big(\A(I_1)'\cap\A(\tilde I)\big)\vee \big(\A(J_1)'\cap\A(\tilde J)\big)
= \A(I)\vee\A(J) = \A_+(\O)
\end{multline*}
because $\A$ is split and strongly additive.

If $X\in\A_+(\O)$ then 
\[
\e_{\tilde\O}(X)\in\pi_+\big(\A_+(\O_1)'\cap\A_+(\tilde\O)\big) = 
\pi_+(\A_+(\O_1))'\cap\pi_+(\A_+(\tilde\O))= 
\pi_+\big(\A_+(\O)),
\]
so $\e_{\tilde\O}$ restricts to an expectation $\B(\O)\to\pi_+(\A_+(\O))$ that is equal to $\e_\O$ by the uniqueness of the expectation.
\end{proof}
With $\B^{(1)}, \B^{(2)}$ local BCFT nets, we shall say that $\B^{(1)}$ is (unitarily) \emph{equivalent} to $\B^{(2)}$ if there is a unitary equivalence of nets on $M_+$ and the unitary interchanges the representations of the corresponding $\A$-subnets.

We shall say that $\B^{(1)}, \B^{(2)}$ are \emph{locally isomorphic} if they are isomorphic as nets on $M_+$ and the isomorphism interchanges the corresponding $\A$-subnets.
\section{Adding a boundary}
\label{adding}
Let $\B_2$ a CFT local net of von Neumann algebras on the Minkowski plane $M$, acting on the Hilbert space $\H_2$, with $U$ the associated covariance positive energy representation of $\bar\G\times\bar\G$ with vacuum unit vector $\Omega$ (see \cite{KL2}).
We assume that $\B_2$ is irreducible (or, equivalently, that, $\Omega$ is the unique $U$-invariant vector up to a phase), so the local von Neumann algebras $\B_2(\O)$, $\O\in\K$, are factors.
Then $\B_2$ is an irreducible extension of the CFT subnet $\A_2$, the tensor product of the chiral subnets that we assume to be equal (parity symmetry), so $\A_2\simeq \A\otimes\A$, namely $\A_2(\O) = \A(I)\otimes\A(J)$ with $I\times J=\O\in\K$ where $\A$ is a local M\"obius covariant net on $\RR$. More precisely, there exists a representation $\pi_2$ of $\A_2$ on $\H_{\B_2}$ such that $\pi_2(\A_2(\O))\subset \B_2(\O)$ is an irreducible inclusion of factors and $U(g)\pi_2(\A_2(\O))U(g)^* = \pi_2(\A_2(g\O))$, $\O\in\K$.
\footnote{$\A_2$ is a canonical subnet of $\B_2$ and its irreducibility follows by diffeomorphism covariance  \cite{KL2}. In this paper only M\"obius covariance is assumed, so  $\pi_2(\A_2(\O))'\cap \B_2(\O)=\mathbb C$ is indeed an assumption.}

A we are assuming $\A$ to be completely rational, also $\B_2$ is completely rational and $\pi_2(\A_2(\O))\subset \B_2(\O)$ has finite index \cite{KLM,KL2}.

We want to construct a BCFT net $\B$ on $M_+$ based on $\A$ in the sense of \cite{LR2}, as specified in Sect. \ref{BCFT}, and isomorphic to the restriction of $\B_2$ to $M_+$. Namely we want to construct a net $\B$ on $M_+$, a representation $\pi$ of $\A$ on $\H_\B$ and a local isomorphism $\Phi$ of $\B$ with $\B_2 |_{M_+}$ such that $\Phi_\O: \pi_+(\A_+(\O))\to \pi_2(\A_2(\O))$, $\O\in\K(M_+)$. Moreover $\Phi$ is to be covariant with respect to  covariance unitary representation of $\bar\G$ associated with $\B$ and the restriction to $\bar\G$ of the covariance unitary representation of $\bar\G\times\bar\G$ associated with $\B$ (where $\bar\G$ is the diagonal subgroup of $\bar\G\times\bar\G$).

As $\K_+$ is not a direct family, our procedure will go through the restriction of $\B_2$ to the wedge $W$. So, denote by $ \gB_2(W)$ the $C^*$-algebra generated by $\B_2(\O)$ as $\O$ runs in $\K(W)$. As $\K(W)$ is a direct family, $\gB_2(W)$ is the induct limit of the $\B_2(\O)$'s and $C^*(\B_2 , W) =\gB(W)$. Analogously, $\gA_+(W)$ is the $C^*$-algebra generated by $\A_2(\O)$ as $\O$ runs in $\K(W)$.

Clearly $\pi_2(\gA_+(W)) \subset \gB_2(W)$ and there is a conditional expectation $\e:\gB_2(W)\to\pi_2(\gA_+(W))$ extending the unique finite index expectation $\e_\O:\B_2(\O)\to\pi_2(\A_+(\O))$ for each $\O\in\K(W)$. By continuity, $\e$ satisfies the Pimsner-Popa bound
\[
\e(X)\geq \l X
\]
for every positive element $X\in\gB_2(W)$, where $\l>0$ is the inverse of the Jones index.

Now consider the BCFT net $\A_+$ on $M_+$ associated with $\A$ acting on $\H_\A$ in Sect. \ref{a+}:
\[
\A_+(\O)\equiv \A(I)\vee\A(J)\ ,
\]
where $\O=I\times J\in\K_+$. 
\begin{theorem}
With $\B_2$ a local, M\"obius covariant net on $M$ as above, there exists a Boundary CFT net $\B$ on $M_+$ based on $\A$ such that $\B$ is $\G$-covariantly locally isomorphic to the restriction $\B_2 |_{M_+}$ of $\B_2$ to $M_+$ as explained above. In particular the inclusions $\pi_+(\A_+(\O))\subset \B(\O)$ and $\pi_2(\A_2(\O))\subset\B_2(\O)$ are isomorphic, $\O\in\K_+$.
%
%In more detail: there exists a $\G$-covariant isomorphism $\Phi: \B |_{M_+}\to \B_+$ such that $\Phi_\O\big(\A(I)\otimes\A(J)\big) = \pi_+\big(\A_+(\O)\big)$, for all $\O\equiv I\times J \in\K_+$, where $\pi$ is a DHR representation of $\A$ on $\H_\B$.
\end{theorem}
\noindent
This theorem will be a consequence of Thm. \ref{BCFTs} that we are going to prove.
\smallskip

\noindent
Let $\Omega$ be the vacuum vector in the Hilbert space of $\H_\A$ and  $\f_0$ the associated vacuum state on $\gA_+(W)$ implemented by $\Omega$:
\[
\f_0(XY) = (\Omega,XY\Omega), \quad X\in \A(I), Y\in\A(J), \quad \O= I\times J\in\K(W).
\]
We now extend $\f_0$ to a state $\f$ on $\gB_2(W)$ by composing it with the conditional expectation $\e$:
\[
\f\equiv\f_0\cdot\e\  .
\]
As $\gA_+(W)$ is a simple $C^*$-algebra by Lemma \ref{simple}, $\pi_2$ is one-to-one on $\gA_+(W)$; to simplify the notations we are identifying $\gA_+(W)$ with $\pi_2(\gA_+(W))$ here and in the following.
\begin{lemma}\label{fmirrep} The GNS representation $\pi_\f$ of $\f$ decomposes into finitely many irreducible representations of $\gB_2(W)$. Indeed $\pi_\f |_{\gA_+(W)}$ decomposes into finitely many irreducible representations of $\gA_+(W)$.
\end{lemma}
\begin{proof} 
By Lemma \ref{simple},  $\gA_+(W)$  and $\gB(W)$ are simple, purely infinite $C^*$-algebras in Cuntz standar form. So the lemma follows by a result by Izumi \cite[Lemma 5.2]{I} because $\e$ has finite index and $\f_0$ is a pure state by Lemma \ref{irrW}.
\end{proof}
Denote by $\B_W$ the restriction of $\B_2$ to $W$.
\begin{lemma}\label{ext2}
$\pi_\f$ gives a covariant (reducible) representation of $\B_W$, hence a covariant representation $\tilde\pi_\f$ of $\B_2 |_{M_+}$ (by Prop. \ref{extpi}). 
\end{lemma}
\begin{proof}
By Lemma \ref{Wcov}, the covariance of $\pi_\f$ will follow once we check that $\f$ is $\a$-invariant and bicyclic for $\B(\O)\subset\B(\tilde\O)$ if $\O\subset\tilde\O\in\K(W)$. Here $\alpha$ is the local action on $\B_2(W)$ given by the covariance unitary representation of $\bar\G$ for $\B_2$ as in eq. \eqref{invariant} ($\bar\G$ the diagonal subgroup of $\bar\G\times\bar\G$).

Now $\f_0$ is $\a$-invariant and bicyclic for $\A_+(\O)\subset\A_+(\tilde\O)$, indeed in the GNS representation of $\f_0$ the vector $\xi_{\f_0}$ is the vacuum vector for $\A$ which is invariant with the Reeh-Schlieder property.

Now the conditional expectation $\e$ restricts to the unique expectation $\e_\O:\B_2(\O)\to\pi_2(\A_+(\O))$, so $\e^{g\O}\cdot\a_g^\O = \a_g^\O\e^\O$, namely $\e$ is $\a$-invariant, therefore $\f$ is $\a$-invariant. As the index of $\pi_2(\A_+(\O))\subset\B_2(\O)$ is independent of $\O$, $\f$ is then bicyclic for $\B_2(\O)\subset\B_2(\tilde\O)$ by Lemma \ref{bicyclic}.
\end{proof}
\begin{lemma}\label{res}
The restriction of $\tilde\pi_\f$ to $\A_+$ is equivalent to $\s_+$ where $\s$ is a DHR representation of $\A$.
\end{lemma}
\begin{proof}
Clearly the restriction of $\pi_\f$ to $\gA_+(W)$ contains $\pi_{\f_0}$ and $\pi_{\f_0}$ gives, by covariance, the representation $\pi_+$ of $\A_+$, where $\pi$ is the vacuum representation of $\A$. So $\tilde\pi_\f$ restricts to $\pi_+$ on $\A_+$.

Thus we can apply Prop. \ref{repA_+} with $\Omega\equiv\xi_\f\in\H_{\B_2}$, indeed $\xi_\f$ implements $\f_0$ on $\A_+$, so it has the needed bicyclicity properties. The only point to check is that $\xi_\f$ is separating for $W_1$; but this is true by Prop. \ref{sep}.
\end{proof}
\begin{lemma}\label{commut}
Let $\tilde\pi_\f$ be representation of the net $\B_2 |_{M_+}$ on $M_+$ obtained by Lemma \ref{ext2}. Then $\pi_\f(\gB_2(W))'' = \tilde\pi_\f(C^*(\B_2,M_+))''\ \Big(\!= \big(\bigcup_{\O\in\K_+}{\tilde\pi}_\f(\B_2(\O)\big)''\Big)$.
\end{lemma}
\begin{proof}
By Lemma \ref{fmirrep} the restriction of $\pi_\f$ to $\gA_+(W)$ is the direct sum of finitely many irreducible representations, so $\pi_\f(\gA_+(W))'$ is finite dimensional.

Let $U$ be the unitary representation of $\bar\G$ implementing the covariance of $\tilde\pi_\f$ on $\H_\f$. Then $U$ implements the covariance of the representation $\s$ of $\A$ on $\H_\f$, given by Lemma \ref{res}, such that $\s_+ =\tilde\pi_\f$.

Now, by Lemma \ref{irrW}, $\s$ is the direct sum of finitely many irreducible representations of $\A$, so $\s$ is covariant with positive energy, moreover there exists a unique unitary representation $V$ of $\G$ implementing the covariance of $\s$ (because $V(g)\in\s(\gA)''$), therefore $U=V$ has positive energy and belongs to $\s_+(\gA_+(M_+))''$, see \cite[Prop. 2.2.]{GL2}.

If $X$ belongs to $\gB_2(W)'$, then $X\in\s_+(\gA_+(W))'=\s(\gA)'$, hence $X$ commutes with $U$. So $X$ commutes with $U(g)\gB_2(W)U(g)^*$, $g\in\bar\G$, hence with $\gB_2(M_+)$ as desired.
\end{proof}
During the proof of Lemma \ref{commut} we have proved the following.
\begin{corollary}\label{posen}
Let $U$ be the unitary representation of $\bar\G$ implementing the covariance of $\tilde\pi_\f$ on $\H_\f$, and $\s$ the representation of $\A$ on $\H_\f$ given by Lemma \ref{res}. Then $U$ has positive energy and $U(g)\in\s_+(\gA(W))''$.
\end{corollary}
\begin{theorem}\label{BCFTs}
$\tilde\pi_\f$ is the direct sum of finitely many irreducible representations of $\B_2 |_{M_+}$.

Each irreducible component of $\tilde\pi_\f$ is covariant and defines an irreducible, local Boundary CFT net on $M_+$.
\end{theorem}
\begin{proof} The first statement is immediate from Lemmas \ref{fmirrep} and \ref{commut}.

$\tilde\pi_\f$ is covariant, see Lemma \ref{ext2}, and let $U$ be the covariance unitary representation of $\G$; as the center of $\tilde\pi_\f(C^*(\B))''$ is finite-dimensional and globally Ad$U$-invariant, it is indeed Ad$U$ pointwise fixed because $\G$ is connected, so each irreducible component of $\pi_\f$ is covariant.

We have only to check the positivity of the energy. But this the content of Corollary \ref{posen}.
\end{proof}
\section{The local isomorphism class of a BCFT net}
Given a local, $\bar\G\times\bar\G$-covariant net of factors $\B_2$ on $M$ with chiral subnet isomorphic to $\A\otimes\A$, in Section \ref{adding} we have shown how to construct finitely many boundary CFT nets $\B^{(i)}$ on $M_+$ based on $\A$
\be\label{morita}
\B^{(i)}(\O) \equiv \tilde\pi_{\f,i}(\B_2(\O)),\quad \O\in\K_+,\ i=1,\dots n\ ,
\ee
where $\tilde\pi_\f = \tilde\pi_{\f,1}\oplus\tilde\pi_{\f,2}\oplus\cdots\oplus\tilde\pi_{\f,n}$ and each $\tilde\pi_{\f,i}$ is an irreducible representation of $C^*(\B_2 , M_+)$.

The $\B^{(i)}$ are all isomorphic as QFT nets on $M_+$, so all locally isomorphic, indeed each $\B^{(i)}$ isomorphic to the restriction of $\B_2$ to $M_+$, and the isomorphism interchanges the representation of $\A_+$. Motivated by the algebraic and tensor categorical description of BCFT, we shall say the two BCFT nets on $M_+$ are \emph{Morita equivalent} if they are locally isomorphic as QFT nets on $M_+$ and the isomorphism interchanges the $\A_+$ subnets. The $\B^{(i)}$'s are Morita equivalent since the isomorphism of of $\B^{(i)}$ with $\B |_{M_+}$ carries the $\A_+$-subnet onto $\big(\A\otimes\A\big)|_{M_+}$. 

In this Section we show that the Morita equivalence class is complete, namely every irreducible BCFT net Morita equivalent to $\B |_{M_+}$ is unitarily equivalent to one of the $\B^{(i)}$'s.

The fact that every BCFT net arise in this way, by restriction of a CFT net on $M_+$, has been shown in \cite{LR2}.

So let $\B$ a BCFT net on $M_+$ based on the chiral (completely rational) local conformal net $\A$, and Morita equivalent to $\B_2 |_{M_+}$ (in a similar sense as above). Let $\Omega_\B\in\H_\B$ be the vacuum vector of $\B$ and $\o$ the associated state of $C^*(\B,M)$. 

Notice that the construction of the representation $\tilde\pi_\f$ of $C^*(\B_2,M_+)$ only depends on the QFT net $\B_2 |_{M_+}$ on $M_+$ and its $\A\otimes\A$-subnet. Therefore all we have to show is that the representation $\pi_\o$ of $C^*(\B,M)$ is contained in $\tilde\pi_\f$ of $C^*(\B,M)$, where $\pi_\o$ is the identity representation of $\B$, i.e. the GNS representation associated with the vacuum state $\o$ of $\B$.

Now the state $\f$ of $\gB(W)$ is invariant under the the conditional expectation: $\f = \f_0\cdot \e$ and $\o$ restricts to $\f_0$ on $\gA_+(W)$, therefore
\[
\f(X) = \f_0( \e(X)) = \o(\e(X)) \geq \l \o(X),
\]
for all positive $X\in \gB(W)$, due to the Pimsner-Popa inequality $\e(X)\geq \l X$, with $\l$ the inverse of the Jones index.

Namely $\o|_{\gB(W)}$ is dominated by $\f$. Thus ${\pi_\o}|_{\gB(W)}$ is the subrepresentation of $\pi_\f$ corresponding to a projection $z\in\gB(W)' = \gB(M_+)'$ by Lemma \ref{commut}. Thus $\pi_\o$ is contained in $\tilde\pi_\f$.

So we have:
\begin{theorem}
The construction in Theorem \ref{BCFTs} gives all BCFT nets on $M_+$ that are covariantly locally isomorphic to $\B_2 |_{M_+}$. 

Indeed the BCFT nets $\B^{(i)}$ in \eqref{morita} are Morita equivalent and they exhaust the Morita equivalence class.
\end{theorem}
\noindent
By the same construction, we have also implicitly shown that, starting with a BCFT net $\B$ on $M_+$, all BCFT nets on $M_+$ locally isomorphic to $\B$ are obtained as above.

\section*{Outlook}
It would be interesting to extend and analyze our results in various directions.

First of all, an extension of our construction in the setting of \cite{LW} could lead to new boundary QFT models where one starts with two chiral nets. Related to this point would also be the analysis in the contexts of thermal states and different boundary regions, cf. \cite{LR4,CLTW1,CLTW2}. 

A natural problem is to classify the (finitely many) Boundary CFT nets that are obtained by adding the boundary to a given completely rational, local conformal net on the Minkowski plane. As described in \cite{LR1}, this is related to the computation of the DHR orbits of an extension of a corresponding local net $\A$ on the boundary real line. One simple example were this can be illustrated is the Ising model $\A$: in this case there are only two irreducible chiral extensions: the identity $\A \supset \A$ and the Fermi extensions $\F \supset \A$. Hence, it follows from \cite{LR1} that there are exactly (up to equivalence) two Haag-dual BCFT nets based on $\A$: the dual net $\A_{+}^{\rm dual}$ of $\A_+$ and  a net $\B$ 
whose restriction to the boundary is $\F$ (see \cite{LR1} for a description). The two chiral extensions $\A$ and $\F$ belong to the same DHR orbit (cf. \cite[page 950]{LR1}) so $\A_{+}^{\rm dual}$ and $\B$ are both locally equivalent to the unique maximal two-dimensional extension $\A_2$ of 
$\A \otimes \A$ (the Longo-Rehren extension), see also \cite{KL2}. Now let $\tilde{\B}$ be any other $BCFT$ based on $\A$ which is locally isomorphic to $\A_2$ and let $\tilde{\B}^{\rm dual}$ the corresponding dual BCFT net. Then, either $\tilde{\B}^{\rm dual}= \B$  or 
$\tilde{\B}^{\rm dual}=\A_{+}^{\rm dual}$. In the first case $\tilde{\B}$ is intermediate between  $\A_+$ and $\B$ and the local isomorphism implies that 
$\tilde{\B}=\B$ as they have the same index. Similarly, in the second case, we have  $\tilde{\B}=\A_{+}^{\rm dual}$. Hence $\A_{+}^{\rm dual}$ and $\B$ are the only BCFT nets based on $\A$ and locally equivalent to $\A_2$. 

Another natural problem is to further relate our scheme with the tensor categorial viewpoint, in particular
to the analysis of conformal defects, see \cite{FFRS}.
\bigskip

\noindent
{\bf Aknwoledgments.} We thank  K.-H. Rehren for useful comments.

\end{document}